\documentclass[sigconf]{acmart}

\AtBeginDocument{%
  \providecommand\BibTeX{{%
    \normalfont B\kern-0.5em{\scshape i\kern-0.25em b}\kern-0.8em\TeX}}}

\setcopyright{acmlicensed}
\copyrightyear{2024}
\acmYear{2024}
\copyrightyear{2024}
\acmYear{2024}
\setcopyright{acmlicensed}
\acmConference[KDD '24] {Proceedings of the 30th ACM SIGKDD Conference on Knowledge Discovery and Data Mining }{August 25--29, 2024}{Barcelona, Spain.}
\acmBooktitle{Proceedings of the 30th ACM SIGKDD Conference on Knowledge Discovery and Data Mining (KDD '24), August 25--29, 2024, Barcelona, Spain}
\acmISBN{979-8-4007-0490-1/24/08}
\acmDOI{10.1145/3637528.3671856}

\usepackage[utf8]{inputenc} 
\usepackage[T1]{fontenc}    
\usepackage{hyperref}       
\usepackage{url}            
\usepackage{booktabs}       
\usepackage{amsmath}
\usepackage{amsfonts}       
\usepackage{nicefrac}       
\usepackage{microtype}      
\usepackage{xcolor}         
\usepackage{ulem}
\usepackage{graphicx}
\usepackage{amsthm}

\usepackage{amssymb}
\usepackage{siunitx}
\usepackage{stfloats}
\usepackage{subfigure}
\usepackage{hyperref}
\usepackage{enumerate}
\usepackage{enumitem}

\usepackage{bm}
\usepackage{nicefrac}
\usepackage{booktabs}
\usepackage{array}
\usepackage{multirow}
\usepackage{threeparttable}
\usepackage{makecell}
\usepackage{tabularx}
\usepackage[procnumbered,ruled,vlined,linesnumbered]{algorithm2e}

\def\eps{\epsilon}

\def\trace#1{\mathrm{Tr} \left(#1 \right)}

\newtheorem{problem}{Problem}
\newtheorem{theorem}{Theorem}[section]

\newtheorem{lemma}[theorem]{Lemma}

\newtheorem{definition}[theorem]{Definition}

\def\calG{\mathcal{G}}

\def\calN{\mathcal{N}}

\newcommand{\removelatexerror}{\let\@latex@error\@gobble}

\newcommand\LL{\bm{\mathit{L}}}

\newcommand{\samv}{{\widetilde{V}}}

\def\trace#1{\mathrm{Tr} \left(#1 \right)}

\newcommand{\STW}[0]{\texttt{STW}\xspace}
\newcommand{\SWF}[0]{\texttt{SWF}\xspace}
\newcommand{\SNB}[0]{\texttt{SNB}\xspace}
\newcommand{\SNBp}[0]{\texttt{SNB+}\xspace}

\newcommand{\one}{\mathbf{1}}

\newcommand\xx{\boldsymbol{\mathit{x}}}

\newcommand\bb{\boldsymbol{\mathit{b}}}

\newcommand\dd{\boldsymbol{\mathit{d}}}
\newcommand\ee{\boldsymbol{\mathit{e}}}

\newcommand\hh{\boldsymbol{\mathit{h}}}
\newcommand\ww{\boldsymbol{\mathit{w}}}

\renewcommand\AA{\boldsymbol{\mathit{A}}}

\newcommand\JJ{\boldsymbol{\mathit{J}}}
\newcommand\DD{\boldsymbol{\mathit{D}}}

\newcommand\PP{\boldsymbol{\mathit{P}}}
\newcommand\MM{\boldsymbol{\mathit{M}}}

\newcommand\QQ{\boldsymbol{\mathit{Q}}}

\newcommand\II{\boldsymbol{\mathit{I}}}
\newcommand\uu{\boldsymbol{\mathit{u}}}

\DontPrintSemicolon
\SetKw{KwAnd}{and}
\SetFuncSty{textsc}
\SetKwInOut{Input}{Input\ \ \ \ }
\SetKwInOut{Output}{Output}
\begin{document}

\title{Fast Query of Biharmonic Distance in Networks}

\author{Changan Liu}
\affiliation{%
  \institution{Fudan University}
  \city{Shanghai}
  \country{China}
  \postcode{200433}
}
\email{19110240031@fudan.edu.cn}

\author{Ahad N. Zehmakan}
\affiliation{%
	\institution{Australian
National University}
	\city{Canberra}
	\country{Australia}}
\email{ahadn.zehmakan@anu.edu.au
}

\author{Zhongzhi Zhang\footnotemark}
\affiliation{
 \institution{Fudan University}
 \city{Shanghai}
 \country{China}
}
\email{zhangzz@fudan.edu.cn}

\renewcommand{\shortauthors}{Changan Liu, Ahad N. Zehmakan, \& Zhongzhi Zhang}

\begin{abstract}

The \textit{biharmonic distance} (BD) is a fundamental metric that measures the distance of two nodes in a graph. It has found applications in network coherence, machine learning, and computational graphics, among others. In spite of BD's importance, efficient algorithms for the exact computation or approximation of this metric on large graphs remain notably absent. In this work, we provide several algorithms to estimate BD, building on a novel formulation of this metric. These algorithms enjoy locality property (that is, they only read a small portion of the input graph) and at the same time possess provable performance guarantees. In particular, our main algorithms approximate the BD between any node pair with an arbitrarily small additive error $\eps$ in time $O(\frac{1}{\eps^2}\text{poly}(\log\frac{n}{\eps} ))$. Furthermore, we perform an extensive empirical study on several benchmark networks, validating the performance and accuracy of our algorithms.
\end{abstract}



\maketitle

\renewcommand{\thefootnote}{*}
\footnotetext[1]{Corresponding author. Changan Liu and Zhongzhi Zhang are with  Shanghai Key Laboratory of Intelligent Information Processing, School of Computer Science, Fudan University, Shanghai 200433, China. This work was supported by the National Natural Science Foundation of China (Nos. U20B2051,  62372112, and 61872093).}

\section{Introduction}
Graph distance metrics have long been a cornerstone in the realm of network analysis, with a plethora of measures being introduced over the years, finding applications across diverse domains~\cite{simibipar,lu2011link,shimada2016graph,spdwww20,8970897,sdwww19,prunwww19}. A prominent example is the geodesic distance~\cite{newman2018networks}, which, intuitively speaking, gauges distance based on the shortest path between node pairs. Another notable metric is the resistance distance, rooted in electrical circuit theory, which has found significance in areas ranging from circuit theory \cite{doyle1984random} and chemistry \cite{klein1993resistance,peng2017kirchhoff} to combinatorial matrix \cite{bapat2010graphs,yang2013recursion} and spectral graph theory \cite{chen2007resistance}.

One of the most fundamental and important distance measures which has gained substantial attention, due to its various applications~\cite{ACC2018BHD,Yi2022BiharmonicDP,verma2017hunt,kreuzer2021rethinking,black2023understanding,YiSh2018,Tyloo2017RobustnessOS,fan2020spectral,bd14,XuWuZhZhKaCh22}, is the \textit{biharmonic distance} (BD). Introduced by Lipman et al.~\cite{lipman2010biharmonic}, BD was originally conceived as a measure to gauge distances on curved surfaces, a pivotal challenge in computer graphics and geometric processing.
The connection between BD and other measures and its different variants have been studied in the prior work. For example,
Yi et al.~\cite{YiSh2018} established a connection between BD and edge centrality. Wei et al.~\cite{wei2021biharmonic} introduced the notation of biharmonic graph distance and explored its relationship with the Kirchhoff index~\cite{klein1993resistance}.

The BD offers a unique blend of both local and global graph structure properties, making it superior to the geodesic and resistance distance in various applications. The versatility of BD is evident, for example, in its application to measure the robustness of second-order noisy consensus problems~\cite{bamieh2012coherence,Yi2022BiharmonicDP}. BD has also found applications in other areas such as graph learning~\cite{kreuzer2021rethinking,black2023understanding}, physics~\cite{Tyloo2017RobustnessOS}, and other network related fields~\cite{fan2020spectral,bd14,tyloo2019key}.

Despite its importance, the computational landscape of BD, especially for expansive graphs, remains challenging. While Yi et al.~\cite{YiSh2018} proposed a promising algorithm for approximating all-pairs BD, the algorithm's dependency on the entire graph input and the construction of a large dense matrix in the preprocessing step makes it inapplicable to modern, rapidly expanding networks. Recognizing this gap, our paper delves deeper into devising fast algorithms for computing and approximating BD. Our contributions can be summarized as follows:

\begin{itemize}[leftmargin=*]
    \item We first develop a novel formula of BD utilizing the expansion of the Laplacian pseudoinverse.
    \item Leveraging this formula, we present two novel local algorithms, \texttt{Push} and \texttt{Push+}, for estimating pairwise BD. To reduce computational complexity, we introduce a random walk-based algorithm, \STW. Our pinnacle contribution is the \SWF algorithm, which employs empirical estimates of the variance to reduce the number of random walks used by \STW.
    \item Based on the pairwise BD estimation algorithms, we also present \SNB, and its more efficient version \SNBp by integrating the fast summation estimation techniques, for estimating nodal BD (which is the sum of the BD between a specific node and other nodes).
    \item Our empirical studies, spanning real networks, underscore the superiority of our algorithms, often outperforming state-of-the-art solutions by significant margins.
\end{itemize}

\section{Preliminary}\label{sec:pre}
\subsection{Notations}
Throughout this paper, we utilize bold lowercase (e.g., $\xx$) for vectors and uppercase (e.g., $\MM$) for matrices. Elements are indicated by subscripts, e.g., $\xx_i$ or $\MM_{i,j}$. For matrices, $\MM[i,:]$ and $\MM[:,j]$ represent the $i$-th row and $j$-th column, respectively. We use $\ee_i$ to denote $i$-th standard basis vector of appropriate dimension, with $i$-th element being 1 and other elements being 0. We use $\mathbf{1}$ (resp. $\JJ$) to denote the all-ones (column) vector (resp. matrix) of proper dimension. We use $\xx^{\top}$ to denote the transpose of vector $\xx$. Furthermore, let $\mathbb{I}_{x=y}$ be an indicator function, which equals 1 if $x=y$ and otherwise 0.

Consider a connected undirected graph (network) $\calG = (V, E)$ with nodes $V$ and edges $E \subseteq V \times V$. Let $n := |V|$ and $m := |E|$ denote the number of nodes and the number of edges, respectively. We use $\mathcal{N}\left(i\right)$ to denote the set of neighbors of $i$, where the degree is $\dd_{i}=\left|\mathcal{N}\left(i\right)\right|$. The Laplacian matrix of $\calG$ is the symmetric matrix $\LL = \DD - \AA$, where $\AA\in\{0,1\}^{n\times n}$ is the adjacency matrix whose entry $\AA_{i,j}=1$ if node $i$ and node $j$ are adjacent, and $\AA_{i,j}=0$ otherwise, and $\DD$ is the diagonal matrix $\DD=\text{diag}(\dd_1,\cdots,\dd_n)$. For any pair of distinct nodes $u,v\in V$, we define $\bb_{uv} = \ee_{u}-\ee_{v}$. Let $\gamma_1 \leq \gamma_2 \leq$ $\gamma_3 \leq \cdots \leq \gamma_n$ be the eigenvalues of $\LL$. It is well-known that $\gamma_1=0$, and $\gamma_2$ is the algebraic connectivity of $\calG$. Matrix $\LL$ is positive semi-definite. Its Moore-Penrose pseudoinverse is $\LL^\dag = \big(\LL +\frac{1}{n}\JJ\big)^{-1}-\frac{1}{n}\JJ$. We use $\widetilde{\LL}=\DD^{-1/2}\LL\DD^{-1/2}$ to denote the normalized Laplacian matrix.

Random walks underpin our proposed algorithms. Let $\PP=\DD^{-1} \AA$ be the random walk matrix (i.e., transition matrix) of $\mathcal{G}$, in which $\PP_{i, j}=\frac{1}{\dd_{i}}$ if $e_{i, j} \in E$ and $\PP_{i, j}=0$ otherwise. Correspondingly, we denote $p_{\ell}\left(i, j\right)=\PP^{\ell}_{i, j}$, which can be interpreted as the probability of a random walk from node $i$ visits node $j$ at the $\ell$-th hop. In this paper, following the convention, we assume $\mathcal{G}$ is not bipartite (which is the case for most real-world networks). According to~\cite{Motwani1995RandomizedA}, the random walks over $\mathcal{G}$ are ergodic, i.e., $\lim _{\ell \rightarrow \infty} \PP^{\ell}_{i, j}=\bm{\pi}_j = \frac{\dd_{j}}{2 m}$ for any $i, j \in V$, where $\bm{\pi}$ denotes the stationary distribution of a random walk starting from any node. Let $\QQ=\DD^{-1/2} \AA \DD^{-1/2}=\DD^{1 / 2} \PP \DD^{-1 / 2}$. Recall that $\LL=\DD-\AA=\DD^{1 / 2}(\II-\QQ) \DD^{1 / 2}$. Note that $\QQ$ is symmetric and is similar to $\PP$. Let $\lambda_1 \geq \lambda_2 \geq$ $\lambda_3 \geq \cdots \geq \lambda_n$ be the eigenvalues of $\QQ$ (and also $\PP$ by the similarity of $\PP$ and $\QQ$ ), with corresponding (row) orthonormal eigenvectors $\uu_1, \uu_2, \ldots, \uu_n$, i.e., $\uu_j \QQ=\lambda_j \uu_j$. Let $\ww_1, \ww_2, \ldots, \ww_n$ be the column eigenvectors of $\PP$, i.e., $\PP\ww_j=\lambda_j\ww_j$. Let $\lambda=\max \left\{\left|\lambda_2\right|,\left|\lambda_n\right|\right\}$. Notably, $\lambda_1=1$, $\uu_1=\frac{\mathbf{1} \DD^{1 / 2}}{\sqrt{2 m}}$ and $\ww_1=\mathbf{1}$~\cite{haveliwala2003second}.

\subsection{Problem Definition}
Since its introduction in~\cite{lipman2010biharmonic}, biharmonic distance in graphs has been extensively studied~\cite{bd14,ACC2018BHD,YiSh2018, Yi2022BiharmonicDP,zhang2020fast,black2023understanding,chen2007resistance,wei2021biharmonic}. Smaller BD implies a closer node connection, while a larger distance suggests a more indirect connection. Its formal definition is as follows.
\begin{definition}\label{def:biharmonic}
    (Pairwise BD~\cite{YiSh2018})
    For a graph $\calG=(V, E)$ with the Laplacian matrix $\LL$, the biharmonic distance $b(s, t)$ between any pair of distinct nodes $s,t\in V$ is defined by
    \begin{align}\label{eq:exact}
        b^2(s,t)=\bb_{st}^{\top}\LL^{2\dag}\bb_{st}=\|\LL^{\dag}\bb_{st}\|^2 = \LL_{s,s}^{2\dag} + \LL_{t,t}^{2\dag} - 2\LL_{s,t}^{2\dag}.
    \end{align}
\end{definition}

Hereafter, we will refer to the square of BD $b^2(s,t)$ as $\beta(s,t)$. To avoid redundancy, we will sometimes refer to $\beta(s,t)$ as BD instead of the square of BD, when it is clear from the context.

According to~\cite{wei2021biharmonic}, the BD between any distinct nodes $s$ and $t$ satisfies $2\gamma_n^{-2}\leq \beta(s,t)\leq 2\gamma_2^{-2}$, implying that the more well-connected the network is, the smaller the difference in BD between distinct node pairs.

Based on the pairwise BD, we can derive the nodal BD as follows.
\begin{definition}\label{def:biharmonic_centrality}
    (Nodal BD~\cite{ACC2018BHD,Yi2022BiharmonicDP})
    For a graph $\calG$, the biharmonic distance of a single node $s\in V$ is defined as:
    \begin{align}
        \beta(s)=\sum_{t\in V\backslash \{s\}}\beta(s,t)=n\LL_{s,s}^{2\dag} + \trace{\LL^{2\dag}}.
    \end{align}
\end{definition}

The exact computation of BD involves inverting the Laplacian matrix $\LL$ of $\calG$ with a time complexity of $O(n^{2.3727})$~\cite{leiserson1994introduction}. A natural resort is to attempt to approximate the BD values. In this paper, we focus on approximating BD values and study both the pairwise biharmonic distance query and the nodal biharmonic distance query problems, which are formally defined as follows.
\begin{problem}
    \label{pro:single-pairbd}
    (Pairwise BD query) Given a graph $\calG$, a pair of nodes $(s,t)$ with $s\neq t$, and an arbitrarily small additive error $\eps$, the problem of pairwise BD query is to find $\beta^\prime(s,t)$ such that:\begin{align}
    |\beta(s,t)-\beta^\prime(s,t)|\leq \eps.
    \end{align}
\end{problem}
\begin{problem}
    \label{def:single-source}
    (Nodal BD query) Given a graph $\calG$, a source node $s$, and an arbitrarily small additive error $\eps$, the nodal BD query problem is to find $\beta^{\prime}(s)$ such that:
    \begin{align}
    |\beta(s)-\beta^\prime(s)|\leq n\eps.
    \end{align}
\end{problem}

Biharmonic distance aligns with effective resistance on a related signed graph~\cite{signed2019}, given that $\LL^{2}$ is the ``repelling'' signed Laplacian matrix~\cite{chen2020spectral}. Thus, calculating biharmonic distance
 mirrors computing effective resistance on this signed graph. This suggests that efficient biharmonic distance computation can enhance understanding of effective resistance on signed graphs~\cite{signed2019}.

\subsection{Existing Algorithms}
Approximating BD involves solving the linear Laplacian system $\LL\xx = \bb$, which has been extensively explored in the theoretical computer science. Despite advancements, the most efficient Laplacian solver takes $\Tilde{O}(m)$ time~\cite{SpTe14,CoKyMiPaJaPeRaXu14,solver2023}, proving expensive for large-scale graphs (with millions of nodes and edges). Based on the Laplacian solver, the authors of~\cite{Yi2022BiharmonicDP} introduced a random projection-based method that permits an approximate biharmonic distance query for any node pair in $\calG$ within $O(\log n)$ time. Yet, its preprocessing step, constructing a $(24\log n/\eps^2)\times n$ matrix for a given $\eps$, requires $\tilde{O}(m/\eps^2)$ time, making it unsuitable for large graphs. In many real-world scenarios, we are only concerned with the biharmonic distances for a limited set of critical node pairs, which makes this preprocessing overhead undesirable. Furthermore, their algorithm requires full knowledge of the underlying graph, which is not feasible in many real-world examples. To overcome these issues, we propose some novel algorithms that only need local information (not the entire graph) and could efficiently answer biharmonic distance queries for a set of node pairs, without unnecessary preprocessing overhead. 
\section{The push algorithm}\label{sec:push1}
In this section, we first present a novel formula of BD and subsequently provide a local algorithm \texttt{Push} which relies on approximating the Laplacian pseudo-inverse $\LL^{\dag}$ of the graph.

\subsection{New Formula for Biharmonic Distance}
\begin{lemma}\label{lem:new_formula}
Let $\hh=\sum_{i=0}^{\infty}\bb_{st}^{\top}\PP^i\DD^{-1}$ where $s,t\in V$ are any two distinct nodes, then\begin{align}
        \beta(s,t) = \|\hh\|_2^2-\frac{1}{n}(\hh\mathbf{1})^2.\label{eq:33}
    \end{align}
\end{lemma}

\begin{proof}
    To prove this lemma, we need two materials. The first is that the pseudoinverse of matrix $\LL$ is related to the normalized Laplacian matrix~\cite{bozzo2013moore} and can be written as
$$
\LL^{\dagger}=\left(\II-\frac{1}{n} \JJ\right) \DD^{-1 / 2} \widetilde{\LL}^{\dagger} \DD^{-1 / 2}\left(\II-\frac{1}{n} \JJ\right).
$$
    Second, we have
    \begin{align}
    \notag&\quad\DD^{-\frac{1}{2}}\widetilde{\LL}^{\dag}\DD^{-\frac{1}{2}}\\
    \notag& =\DD^{-\frac{1}{2}} \sum_{j=2}^n \frac{1}{1-\lambda_j} \uu_j^{\top} \uu_j \DD^{-\frac{1}{2}}=\DD^{-\frac{1}{2}} \sum_{j=2}^n \sum_{i=0}^{\infty} \lambda_j^i \uu_j^{\top} \uu_j \DD^{-\frac{1}{2}} \\
    \notag& =\DD^{-\frac{1}{2}} \sum_{i=0}^{\infty} \sum_{j=2}^n \lambda_j^i \uu_j^{\top} \uu_j \DD^{-\frac{1}{2}}=\DD^{-\frac{1}{2}} \sum_{i=0}^{\infty}\left(\QQ^i-\uu_1^{\top} \uu_1\right) \DD^{-\frac{1}{2}}.
    \end{align}
    Therefore, with the above two materials in hand, 
    \begin{align}
    \notag&\quad\beta(s,t)=b^2(s,t)=\bb_{st}^{\top}\LL^{2\dag}\bb_{st}\\
    \notag&=\bb_{st}^{\top}(\DD^{-\frac{1}{2}}\widetilde{\LL}^{\dag}\DD^{-1}\widetilde{\LL}^{\dag}\DD^{-\frac{1}{2}}-\frac{1}{n}\DD^{-\frac{1}{2}}\widetilde{\LL}^{\dag}\DD^{-\frac{1}{2}}\JJ\DD^{-\frac{1}{2}}\widetilde{\LL}^{\dag}\DD^{-\frac{1}{2}})\bb_{st}\\
    \notag&=\bb_{st}^{\top}\Big(\DD^{-\frac{1}{2}}\sum\limits_{i=0}^{\infty}(\QQ^i-\uu_1^{\top}\uu_1)\DD^{-1}\sum\limits_{i=0}^{\infty}(\QQ^i-\uu_1^{\top}\uu_1)\DD^{-\frac{1}{2}}-\\
    \notag&\quad \frac{1}{n}\DD^{-\frac{1}{2}}\sum\limits_{i=0}^{\infty}(\QQ^i-\uu_1^{\top}\uu_1)\DD^{-\frac{1}{2}}\JJ\DD^{-\frac{1}{2}}\sum\limits_{i=0}^{\infty}(\QQ^i-\uu_1^{\top}\uu_1)\DD^{-\frac{1}{2}}\Big)\bb_{st}\\
    \notag&=\sum\limits_{i=0}^{\infty}\bb_{st}^{\top}\DD^{-\frac{1}{2}}\QQ^{i}\DD^{-1}\sum\limits_{i=0}^{\infty}\QQ^i\DD^{-\frac{1}{2}}\bb_{st}-\\
    \notag&\quad\frac{1}{n}\sum\limits_{i=0}^{\infty}\bb_{st}^{\top}\DD^{-\frac{1}{2}}\QQ^{i}\DD^{-\frac{1}{2}}\JJ\DD^{-\frac{1}{2}}\sum\limits_{i=0}^{\infty}\QQ^i\DD^{-\frac{1}{2}}\bb_{st}\\
    \notag&=\sum\limits_{i=0}^{\infty}\bb_{st}^{\top}\PP^i\DD^{-1}\sum\limits_{i=0}^{\infty}\PP^i\DD^{-1}\bb_{st}-\frac{1}{n}\sum\limits_{i=0}^{\infty}\bb_{st}^{\top}\PP^i\DD^{-1}\mathbf{1}\sum\limits_{i=0}^{\infty}\mathbf{1}^{\top}\PP^i\DD^{-1}\bb_{st},
    \end{align}
    which concludes the proof.
\end{proof}

\subsection{A Universal Bound of Truncated Length}
Given an integer $\ell\in[0,\infty]$, for any two distinct nodes $s,t\in V$, let $\hh^{\ell}=\sum_{i=0}^{\ell-1}\bb_{st}^{\top}\PP^i\DD^{-1}$, then we define 
\begin{align}
    \label{eq:betal}
    \beta^{\ell}(s,t)=\|\hh^{\ell}\|_2^2-\frac{1}{n}(\hh^{\ell}\mathbf{1})^2.
\end{align}
We aim to choose an appropriate $\ell$ to approximate $\beta(s,t)$ by $\beta^{\ell}(s,t)$ with bounded error. In the lemma below, we provide a universal bound of length for any pair of nodes.

\begin{lemma}\label{lem:uni_l}
For any additive error $\eps>0$, it holds that $|\beta(s,t)-\beta^{\ell}(s,t)|\leq \eps/2$, where the truncated length $\ell$ satisfies
\begin{align}\label{eq:uni_l}
    \ell=\left\lceil\frac{\log (12n/(\eps(1-\lambda)^2))}{\log(1/\lambda)}\right\rceil.
\end{align}
\end{lemma}

\begin{proof}
    To aid our analysis, we write $\bb_{st}^{\top}\DD^{-1/2}=\sum_{j=1}^n\alpha_j\uu_j$. Note that $\alpha_1=\bb_{st}^{\top}\DD^{-1/2}\uu_1^{\top}=0$. Thus $\bb_{st}^\top\DD^{-1/2}=\sum_{j=2}^n\alpha_j\uu_j$. Next, we write $\DD^{-1/2}\ee_v=\sum_{j=1}^n\beta_j\uu_j^{\top}$. Then, we have $\bb_{st}^{\top}\DD^{-1/2}\QQ^i\DD^{-1/2}\ee_v=\sum_{j=1}^n\lambda_j^i\alpha_j\beta_j$ and 
   $|\sum_{j=1}^n\alpha_j\beta_j|$$=|\bb_{st}^{\top}\DD^{-1/2}\DD^{-1/2}\ee_v|=|\bb_{st}^{\top}\DD^{-1}\ee_v|\leq|\bb_{st}^{\top}\ee_v|\leq 1$.
   
   Now, we first analyze the error of every element $\hh_v$ for $v\in V$:
   \begin{align*}
   &\quad|\hh_v-\hh_v^{\ell}|=\Big|\sum\limits_{i=\ell}^{\infty}\bb_{st}^{\top}\DD^{-1/2}\QQ^i\DD^{-1/2}\ee_v\Big|\\
   &=\Big|\sum\limits_{i=\ell}^{\infty}\sum\limits_{j=1}^{n}\alpha_j\uu_j\sum\limits_{j=1}^{n}\lambda_j^i\uu_j^{\top}\uu_j\sum\limits_{j=1}^{n}\beta_j\uu_j^{\top}\Big|\\
   &=\Big|\sum\limits_{i=\ell}^{\infty}\sum\limits_{j=1}^{n}\lambda_j^i\alpha_j\beta_j|\leq \sum\limits_{i=\ell}^{\infty}\lambda ^i|\sum\limits_{j=1}^{n}\alpha_j\beta_j\Big|
   \leq \frac{\lambda^\ell}{1-\lambda}.
   \end{align*}

   The above result shows that the larger the $\ell$, the smaller the approximation error is.
   
   Using the similar method, we can also obtain that $\hh_v\leq\frac{1}{1-\lambda}$.
   Then, we can obtain that
   \begin{align}
       \notag&\quad\Big|\|\hh\|_2^2-\|\hh^{\ell}\|_2^2\Big|=\Big|\sum_{v=1}^n(\hh_v^2-(\hh_v^{\ell})^2)\Big|\leq \sum_{v=1}^n|\hh_v^2-(\hh_v^{\ell})^2|\\
       \notag&=\sum_{v=1}^n|\hh_v-\hh_v^{\ell}||2\hh_v+\hh_v^{\ell}-\hh_v|\leq\sum_{v=1}^n|\hh_v-\hh_v^{\ell}|^2+|2\hh_v(\hh_v^{\ell}-\hh_v)|\\
       &\leq \sum_{v=1}^n\frac{\lambda^{2\ell}}{(1-\lambda)^2}+\frac{2\lambda^{\ell}}{(1-\lambda)^2}\leq\frac{3n\lambda^{\ell}}{(1-\lambda)^2},\label{eq:fst}
   \end{align}
   and 
   \begin{align*}
       &|\hh\mathbf{1}-\hh^{\ell}\mathbf{1}|\leq\sum_{v=1}^n|\hh_v-\hh_v^{\ell}|\leq \frac{n\lambda^\ell}{1-\lambda},
   \end{align*}
   and subsequently, we have 
   \begin{align}
       \notag&\quad|(\hh\mathbf{1})^2-(\hh^{\ell}\mathbf{1})^2|=|\hh\mathbf{1}-\hh^{\ell}\mathbf{1}||2\hh\mathbf{1}+\hh^{\ell}\mathbf{1}-\hh\mathbf{1}|\\
       &\leq \frac{n^2\lambda^{2\ell}}{(1-\lambda)^2} + \frac{2n^2\lambda^{\ell}}{(1-\lambda)^2}\leq\frac{3n^2\lambda^{\ell}}{(1-\lambda)^2}.\label{eq:scd}
   \end{align}
   Finally, combining Eq.~\eqref{eq:fst} and Eq.~\eqref{eq:scd} yields
   \begin{align}
       \notag&\big|\|\hh\|_2^2-\frac{1}{n}(\hh\mathbf{1})^2- \big(\|\hh^{\ell}\|_2^2-\frac{1}{n}(\hh^{\ell}\mathbf{1})^2\big)\big|
       \leq \\&\big|\|\hh\|_2^2- \|\hh^{\ell}\|_2^2|
       +\frac{1}{n}|(\hh^{\ell}\mathbf{1})^2-(\hh\mathbf{1})^2\big|\leq\frac{6n\lambda^{\ell}}{(1-\lambda)^2}\leq \eps/2,\notag
   \end{align}
   where the last inequality follows from 
   \begin{align*}
       \ell=\frac{\log (12n/(\eps(1-\lambda)^2))}{\log(1/\lambda)},
   \end{align*}
   concluding the proof.
\end{proof}

\subsection{Algorithm Description} 
Based on the above lemma, we propose an algorithm \texttt{Push} for computing BD. The \texttt{Push} algorithm conducts a deterministic graph traversal from $s$ and $t$ to compute the $i$-hop transition probability values $p_{i}(s,v)$ and $p_{i}(t,v)$ for $0\leq i \leq \ell$ in an iterative manner, as illustrated in Algorithm~\ref{alg:push}. After initial setup, the algorithm computes the truncated length and then undertakes an $\ell$-hop graph traversal originating from $s$ and $t$. Specifically, at the $i$-th hop, it first sets $p_{i}(s, x)=0$ and $p_{i}(t, x)=0$ for any node $x\in V$. Subsequently, for each node $j$ with non-zero $p_{i-1}(s, j)$ and $p_{i-1}(t, j)$ values, it pushes this value among $j$'s neighbors, essentially conducting a sparse matrix-vector multiplication. \texttt{Push} updates $\hh^{\ell}$ during each traversal round and, upon completion, returns the computed $\beta^{\prime}(s,t)= \beta^{\ell}(s,t)$. The following theorem states the worst-case time complexity of \texttt{Push}.
\begin{theorem}
    \label{the:11}
    The algorithm $\texttt{Push}(\calG,\eps,s,t)$ outputs an estimate $\beta^{\prime}(s,t) = \beta^{\ell}(s,t)$ such that $|\beta(s,t)-\beta^{\prime}(s,t)|\leq \eps$, and its run time is in $O(m\log (\frac{n}{\eps}))$ for $\ell=\left\lceil\frac{\log (6n/(\eps(1-\lambda)^2))}{\log(1/\lambda)}\right\rceil$.
\end{theorem}

Notwithstanding its unsatisfying worst-case time complexity, the actual number of graph traversal operations from each node in Algorithm~\ref{alg:push} (Lines 6--9) is far less than $m$ when $\eps$ is not too small.

\normalem
\begin{algorithm}[h]
	\caption{$\texttt{Push}(\calG, \eps,s,t)$}\label{alg:push}
	\Input{
		A connected graph $\calG=(V,E)$ with $n$ nodes; an error parameter $\eps$, two nodes $s$ and $t$
	}
	\Output{
		Estimated biharmonic distance $\beta^{\ell}(s,t)$
	}
	$\ell=Eq.\eqref{eq:uni_l}$; $\hh^{\ell}\gets\mathbf{0}^{1\times n}$\;
    $p_{0}(s,x)\gets0$ $\forall x\in V\backslash\{s\}$; $p_{0}(s,s)\gets1$\;
    $p_{0}(t,x)\gets0$ $\forall x\in V\backslash\{t\}$; $p_{0}(t,t)\gets1$\;
	\For{$i = 1$ to $\ell-1$}{
        $p_{i}(s,x)\gets0$ $\forall x\in V$; $p_{i}(t,x)\gets0$ $\forall x\in V$\;
		\For{$j\in V$ \text{with} $p_{i-1}(s,j)>0$}{
        \lFor{$x\in \calN(j)$}{$p_{i}(s,x)\gets p_{i}(s,x)+\frac{p_{i-1}(s,j)}{\dd_{j}}$}}  
        \For{$j\in V$ \text{with} $p_{i-1}(t,j)>0$}{
        \lFor{$x\in \calN(j)$}{$p_{i}(t,x)\gets p_{i}(t,x)+\frac{p_{i-1}(t,j)}{\dd_{j}}$}
        }  
        $\hh^{\ell}_x\gets \hh^{\ell}_x + \frac{p_{i}(s,x)}{\dd_x} - \frac{p_{i}(t,x)}{\dd_x}\forall x\in V$\;
	}
    $\hh^{\ell}\gets \hh^{\ell}\DD^{-1}$\;
 
    \Return $\beta^{\ell}(s,t)=\|\hh^{\ell}\|_2^2-\frac{1}{n}(\hh^{\ell}\mathbf{1})^2$\;
\end{algorithm}

\section{Refining Maximum Length}\label{sec:push2}
In this section, we introduce an improved estimation for the truncated length, factoring in the structural nuances around the node pair in question. The node-specific \textit{mixing time}, represented as $\zeta_s$ for a source node $s$, reflects the steps a random walk from $s$ needs to stabilize to the graph's stationary distribution, $\bm{\pi}$~\cite{Lang2020DistributedRW}. By setting $\ell=max\{\zeta_s, \zeta_t\}$, we can equate $\beta(s,t)$ to $\beta^{\ell}(s,t)$. Essentially, for a fixed $\ell$, node pairs with lower mixing times tend to approximate $\beta(s,t)$ more accurately. In simpler terms, shorter mixing times for nodes $s,t$ allow a reduced $\ell$ while preserving accuracy. To optimize large $\ell$ values, we propose a costumed $\ell$ for distinct node pairs in the following lemma.

\begin{lemma}\label{lem:cus_l}
For any additive error $\eps>0$, it holds that $|\beta(s,t)-\beta^{\ell}(s,t)|\leq \eps/2$, where the truncated length $\ell$ satisfies
\begin{align}
\label{eq:cus_l}
    \ell=\left\lceil\frac{\log \bigg(\frac{6\sum_{v=1}^n(\frac{1}{\dd_s}+\frac{1}{\dd_t}+\frac{2}{\dd_v})^2+\frac{6}{n}(\frac{n}{\dd_s}+\frac{n}{\dd_t}+\sum_{v=1}^n\frac{2}{\dd_v})^2}{\eps(1-\lambda)^2}\bigg)}{\log(1/\lambda)}\right\rceil.
\end{align}
\end{lemma}
\begin{proof}
    Let us first provide the lemma below.
    
    \begin{lemma}\cite{Yang2023EfficientEO}\label{lem:bound_l}
        Given an undirected graph $\calG$, for any two nodes $s, t \in V$ and any integer $i \geq 1$, the transition matrix $\PP$ can be decomposed as
        \begin{align}
        p_i(s, t)  =\sum_{j=1}^n \ww_{j,s} \ww_{j,t} \bm{\pi}_t \lambda_j^i, \text {and } \frac{1}{\bm{\pi}_t}  =\sum_{j=1}^n \ww_{j,t}^2 \label{eq:2}.
        \end{align}    
    \end{lemma}

    According to Eq.~\eqref{eq:2}, the following equations hold
    \begin{align*}
        \sum_{j=1}^n \ww_{j,s}^2&=\frac{2 m}{\dd_s}, \sum_{j=1}^n \ww_{j,t}^2=\frac{2 m}{\dd_t}, \text { and }\\
        \bb_{st}^{\top}\PP^{i}\DD^{-1}\ee_v&=\frac{p_i(s,v)}{\dd_v}-\frac{p_i(t,v)}{\dd_v}=\frac{1}{2m}\sum_{j=1}^n(\ww_{j,s}-\ww_{j,t})\ww_{j,v}\lambda_j^i.
    \end{align*}
    
    Now, we analyze the approximation error of every element $\hh_v$ for $v\in V$:
        \begin{align}
        \notag|\hh_v-\hh_v^{\ell}|&=\bigg|\sum\limits_{i=\ell}^{\infty}\bb_{st}^{\top}\PP^{i}\DD^{-1}\ee_v\bigg|=\bigg|\frac{1}{2m}\sum_{j=1}^n(\ww_{j,s}-\ww_{j,t})\ww_{j,v}\sum\limits_{i=\ell}^{\infty}\lambda_j^i\bigg|\\
        \notag&\leq \frac{1}{2m}\sum_{j=1}^n\left|(\ww_{j,s}-\ww_{j,t})\ww_{j,v}\right|\sum\limits_{i=\ell}^{\infty}\lambda^i\\
        \notag&\leq\frac{\lambda^{\ell}}{1-\lambda} \frac{1}{2m}\sum_{j=1}^n\left(|\ww_{j,s}\ww_{j,v}|+|\ww_{j,t}\ww_{j,v}|\right)\\
        \notag&\leq \frac{\lambda^{\ell}}{1-\lambda} \frac{1}{2m}\sum_{j=1}^n\left(\ww_{j,s}^2+\ww_{j,v}^2+\ww_{j,t}^2+\ww_{j,v}^2\right)\\
        \notag&=\frac{\lambda^{\ell}}{1-\lambda} \frac{1}{2m} \left(\frac{2m}{\dd_s}+\frac{2m}{\dd_t}+\frac{4m}{\dd_v}\right)\\
        &\notag=\frac{\lambda^{\ell}}{1-\lambda} \left(\frac{1}{\dd_s}+\frac{1}{\dd_t}+\frac{2}{\dd_v}\right).
        \end{align}
    
    Using similar method, we can derive that $$\hh_v\leq \frac{1}{1-\lambda} \left(\frac{1}{\dd_s}+\frac{1}{\dd_t}+\frac{2}{\dd_v}\right).$$
    Subsequently, we derive that 
    \begin{align}
        \notag&\quad\Big|\|\hh\|_2^2-\|\hh^{\ell}\|_2^2\Big|=\Big|\sum_{v=1}^n(\hh_v^2-(\hh_v^{\ell})^2)\Big|\leq \sum_{v=1}^n|\hh_v^2-(\hh_v^{\ell})^2|\\
        \notag&=\sum_{v=1}^n|\hh_v-\hh_v^{\ell}||2\hh_v+\hh_v^{\ell}-\hh_v|\leq\sum_{v=1}^n|\hh_v-\hh_v^{\ell}|^2+|2\hh_v(\hh_v^{\ell}-\hh_v)|\\
            \notag&\leq\frac{3\lambda^{\ell}}{(1-\lambda)^2}\sum_{v=1}^n\Big(\frac{1}{\dd_s}+\frac{1}{\dd_t}+\frac{2}{\dd_v}\Big)^2.
    \end{align}
    For the second component, we first note that 
    \begin{align*}
        &|\hh\mathbf{1}-\hh^{\ell}\mathbf{1}|\leq\sum_{v=1}^n|\hh_v-\hh_v^{\ell}|\leq \frac{\lambda^{\ell}}{1-\lambda}\sum_{v=1}^n\Big(\frac{1}{\dd_s}+\frac{1}{\dd_t}+\frac{2}{\dd_v}\Big),
    \end{align*}
    and $|\hh\mathbf{1}| \leq \frac{1}{1-\lambda}\sum_{v=1}^n\Big(\frac{1}{\dd_s}+\frac{1}{\dd_t}+\frac{2}{\dd_v}\Big).$
    Subsequently, we have
    \begin{align*}
        &|(\hh\mathbf{1})^2-(\hh^{\ell}\mathbf{1})^2|\leq \frac{3\lambda^{\ell}}{(1-\lambda)^2}\Big(\frac{n}{\dd_s}+\frac{n}{\dd_t}+\sum_{v=1}^n\frac{2}{\dd_v}\Big)^2.
    \end{align*}
    Finally,
    \begin{align}
            \notag&\Big|\|\hh\|_2^2-\frac{1}{n}(\hh\mathbf{1})^2- (\|\hh^{\ell}\|_2^2-\frac{1}{n}(\hh^{\ell}\mathbf{1})^2)\Big| \leq\\
            &
            \notag\frac{3\lambda^{\ell}}{(1-\lambda)^2}\Big(\sum_{v=1}^n\Big(\frac{1}{\dd_s}+\frac{1}{\dd_t}+\frac{2}{\dd_v}\Big)^2+\frac{1}{n}\Big(\frac{n}{\dd_s}+\frac{n}{\dd_t}+\sum_{v=1}^n\frac{2}{\dd_v}\Big)^2\Big)\leq \eps/2,
    \end{align}
    where the last inequality follows from 
    \begin{align}
        \notag\ell=\frac{\log \bigg(\frac{6\sum_{v=1}^n(\frac{1}{\dd_s}+\frac{1}{\dd_t}+\frac{2}{\dd_v})^2+\frac{6}{n}(\frac{n}{\dd_s}+\frac{n}{\dd_t}+\sum_{v=1}^n\frac{2}{\dd_v})^2}{\eps(1-\lambda)^2}\bigg)}{\log(1/\lambda)},
    \end{align}
    concluding the proof.
    \end{proof}    

To distinguish between the first universal truncated length $\ell$ and the second customized length $\ell$ for a specified pair of nodes, hereafter we refer to them as $\ell$ and $\ell_{s,t}$, respectively.

\textbf{The \texttt{Push+} Algorithm}. Replacing $\ell$ in Line 1 of Algorithm~\ref{alg:push} with $\ell_{s,t}$ yields a new algorithm \texttt{Push}+. It's crucial to note that in sparser networks with low-degree querying nodes, $\ell_{s,t}$ can exceed $\ell$. The choice between them depends on both local and global network structures. Specifically, $\ell_{s,t}$ suits high-degree nodes in dense networks, while $\ell$ is preferred for low-degree nodes in sparser contexts. Conveniently, a simple computation lets us choose the smaller one.

\section{Approximating by sampling random walks}\label{sec:walk1}

\texttt{Push} and \texttt{Push+} conduct deterministic graph traversal over $\calG$ by sequentially visiting all nodes reachable from $s$ and $t$. Given the small-world property common to social and information networks, most nodes can be quickly reached from one another. As an illustration, by 2012, 92\% of reachable pairs on Facebook were within five steps~\cite{Backstrom2011FourDO}. Thus, when $\ell$ is large and nodes $s, t$ have multiple neighbors, vectors $\PP^{i-1}[s,:]$ and $\PP^{i-1}[t,:]$ become dense quickly. Consequently, each round of pushing can take up to $O(m)$ time, making this approach inefficient for large networks.

To alleviate the efficiency and scalability issue of \texttt{Push} and \texttt{Push+}, while maintaining their superiority in effectiveness, we propose estimating biharmonic distance by sampling random walks. 

\subsection{Approximating by Sampling Random Walks}

Recall that $\beta^{\ell}(s,t)=\|\hh^{\ell}\|_2^2-\frac{(\hh^{\ell}\mathbf{1})^2}{n}$ where $\hh^{\ell}=\sum_{i=0}^{\ell-1}\bb_{st}^{\top}\PP^i\DD^{-1}$. We next resort to random samples to estimate the first and second part of $\beta^{\ell}(s,t)$. Let's first analyze how to approximate the first part. Define $\hh^{\ell,x}=\bb_{st}^{\top}\PP^x\DD^{-1}$. Then, we have
\begin{align}
    \notag &\quad \|\hh^{\ell}\|_2^2=\sum_{v=1}^n \sum_{i=0}^{\ell-1}\hh_{v}^{\ell,i}\sum_{j=0}^{\ell-1}\hh_{v}^{\ell,j}=\sum_{i=0}^{\ell-1}\sum_{j=0}^{\ell-1}\sum_{v=1}^{n}\hh_{v}^{\ell,i}\hh_{v}^{\ell,j}\\
    \notag&=\sum_{i=0}^{\ell-1}\sum_{j=0}^{\ell-1}\sum_{v=1}^{n}\Big(\frac{\PP_{s,v}^{i}\PP_{s,v}^{j}}{\dd_v^2}+\frac{\PP_{t,v}^{i}\PP_{t,v}^{j}}{\dd_v^2}-\frac{\PP_{t,v}^{i}\PP_{s,v}^{j}}{\dd_v^2}-\frac{\PP_{s,v}^{i}\PP_{t,v}^{j}}{\dd_v^2}\Big).
\end{align}

The quantity $\sum_{v=1}^{n}\frac{\PP^{i}_{s,v}\PP_{s,v}^{j}}{\dd_v^2}$ represents the degree-normalized probability of two random walks with length $i$ and $j$ respectively, both starting from $s$, ending at the same node. Define $W_{i,j}=\sum_{1\leq t\leq r}W_t$, where $W_t= 1$ if two walks of length $i$ and $j$ from $s$ end at the same node and this is accepted with a probability of the inverse square of the end node's degree; otherwise, $W_t= 0$. In the same manner, $X_{i,j}$ pertains to two walks from $t$, $Y_{i,j}$ to two walks from $s$ and $t$, and $Z_{i,j}$ to two walks from $t$ and $s$.

We make use of the following Chernoff-Hoeffding's inequality.
\begin{theorem}
  \label{c-h thm}
  (\textsc{Chernoff-Hoeffding's Inequality}~\cite{hoeffding}). Let $Z_1, Z_2, \ldots, Z_{n_z}$ be i.i.d. random variables with $Z_j\left(\forall 1 \leq j \leq n_z\right)$ being strictly bounded by the interval $\left[a_j, b_j\right]$. We define the mean of these variables by $Z=\frac{1}{n_z} \sum_{j=1}^{n_z} Z_j$. Then, we have
  $$
  \mathbb{P}[|Z-\mathbb{E}[Z]| \geq \varepsilon] \leq 2 \exp \bigg(-\frac{2 n_z^2 \varepsilon^2}{\sum_{j=1}^{n_z}\left(b_j-a_j\right)^2}\bigg).
  $$
\end{theorem}
It follows that
$$
\begin{aligned}
& \quad \mathbb{P}\bigg[\bigg|\frac{W_{i, j}}{r}-\sum_{v=1}^{n}\frac{\PP_{s,v}^{i}\PP_{s,v}^{j}}{\dd_v^2}\bigg| \geq \frac{\eps}{16 \ell^2}\bigg] =\mathbb{P}\bigg[\bigg|\frac{W_{i, j}}{r}-\frac{\mathbb{E}\left[W_{i, j}\right]}{r}\bigg| \geq \frac{\eps}{16 \ell^2}\bigg] \\
& \leq 2 \exp (-2 \eps^2 r^2 /(256 \ell^4 r)) \leq \frac{\delta}{4\ell^2}, 
\end{aligned}
$$
with the last inequality due to $r=128 \ell^4\log (\frac{8 \ell^2}{\delta}) / \eps^2$. Similarly, 

\begin{align*}
& \mathbb{P}\bigg[\bigg|\frac{X_{i, j}}{r}-\sum_{v=1}^{n}\frac{\PP_{t,v}^{i}\PP^{j}_{t,v}}{\dd_v^2}\bigg| \geq \frac{\eps}{16 \ell^2}\bigg]\leq \frac{\delta}{4\ell^2}, 
\end{align*}

\begin{align*}
& \mathbb{P}\bigg[\bigg|\frac{Y_{i, j}}{r}-\sum_{v=1}^{n}\frac{\PP^{i}_{t,v}\PP^{j}_{s,v}}{\dd_v^2}\bigg| \geq \frac{\eps}{16\ell^2}\bigg]\leq \frac{\delta}{4\ell^2}, 
\end{align*}

\begin{align*}
& \mathbb{P}\bigg[\bigg|\frac{Z_{i, j}}{r}-\sum_{v=1}^{n}\frac{\PP^{i}_{s,v}\PP^{j}_{t,v}}{\dd_v^2}\bigg| \geq \frac{\eps}{16 \ell^2}\bigg]\leq \frac{\delta}{4\ell^2}, 
\end{align*}

Thus, by a union bound, it holds with a success probability of at least $1-4*\ell^2*\frac{\delta}{4\ell^2} = 1-\delta$ for some $\delta\in (0,1)$ that 
\begin{align*}
    \bigg|\sum_{i=0}^{\ell-1}\sum_{j=0}^{\ell-1}\Big(\frac{W_{i,j}}{r}+\frac{X_{i,j}}{r}-\frac{Y_{i,j}}{r}-\frac{Z_{i,j}}{r}\Big)-\|\hh^{\ell}\|_2^2\bigg|\leq\frac{4\eps\ell^2}{16\ell^2}=\frac{\eps}{4}.
\end{align*}

Let's explore how random walk samples can be employed to approximate the second part which could easily be rewritten as:
\begin{align*}
&\quad \frac{1}{n}(\hh^{\ell}\mathbf{1})^2=\frac{1}{n}\sum_{u=1}^{n}\sum_{i=0}^{\ell-1}\bb_{st}^{\top}\PP^{i}\DD^{-1}\ee_u\sum_{v=1}^{n}\sum_{j=0}^{\ell-1}\bb_{st}^{\top}\PP^{j}\DD^{-1}\ee_v\\
    &=\frac{1}{n}\sum_{i=0}^{\ell-1}\sum_{j=0}^{\ell-1}\sum_{u=1}^n\sum_{v=1}^n\Big(\frac{\PP^{i}_{s,u}\PP^{j}_{s,v}}{\dd_{u}\dd_{v}}+\frac{\PP^{i}_{t,u}\PP^{j}_{t,v}}{\dd_{u}\dd_{v}}-\frac{\PP^{i}_{s,u}\PP^{j}_{t,v}}{\dd_{u}\dd_{v}}-\frac{\PP^{i}_{t,u}\PP^{j}_{s,v}}{\dd_{u}\dd_{v}}\Big).
\end{align*}

Let $\bar{W}_{i,j}=\sum_{1\leq t\leq r}\bar{W}_t$, where $\bar{W}_t$ is an indicator random variable and $\bar{W}_t= 1$ if two walks of length $i$ and $j$ start from $s$ and this is accepted with a probability of the inverse of the degree product of the end nodes; otherwise, $\bar{W}_t= 0$. In the same manner, $\bar{X}_{i,j}$ pertains to two walks from $t$, $\bar{Y}_{i,j}$ to two walks from $s$ and $t$, and $\bar{Z}_{i,j}$ to two walks from $t$ and $s$.

Based on the Chernoff-Hoeffding inequality,
\begin{align*}
& \quad \mathbb{P}\bigg[\bigg|\frac{\bar{W}_{i, j}}{nr}-\sum_{u=1}^{n}\sum_{v=1}^{n}\frac{\PP^{i}_{s,u}\PP^{j}_{s,v}}{n\dd_{u}\dd_{v}}\bigg| \geq \frac{\eps}{16 \ell^2}\bigg]=\mathbb{P}\bigg[\bigg|\frac{\bar{W}_{i, j}}{nr}-\frac{\mathbb{E}[\bar{W}_{i, j}]}{nr}\bigg| \geq \frac{\eps}{16 \ell^2}\bigg] \\
& \leq 2 \exp (-2 \eps^2 r^2n^2 /(256 \ell^4 r)) \leq \frac{\delta}{4 \ell^2},
\end{align*}
where $r=128 \ell^4\log (\frac{8 \ell^2}{\delta}) / (n^2\eps^2)$. Similarly,
\begin{align*}
& \mathbb{P}\bigg[\bigg|\frac{\bar{X}_{i, j}}{nr}-\sum_{u=1}^{n}\sum_{v=1}^{n}\frac{\PP^{i}_{t,u}\PP^{j}_{t,v}}{n\dd_{u}\dd_{v}}\bigg| \geq \frac{\eps}{16 \ell^2}\bigg]\leq \frac{\delta}{4 \ell^2},
\end{align*}
\begin{align*}
& \mathbb{P}\bigg[\bigg|\frac{\bar{Y}_{i, j}}{nr}-\sum_{u=1}^{n}\sum_{v=1}^{n}\frac{\PP^{i}_{s,u}\PP^{j}_{t,v}}{n\dd_{u}\dd_{v}}\bigg| \geq \frac{\eps}{16 \ell^2}\bigg]\leq \frac{\delta}{4 \ell^2},
\end{align*}
\begin{align*}
& \mathbb{P}\bigg[\bigg|\frac{\bar{Z}_{i, j}}{nr}-\sum_{u=1}^{n}\sum_{v=1}^{n}\frac{\PP^{i}_{t,u}\PP^{j}_{s,v}}{n\dd_{u}\dd_{v}}\bigg| \geq \frac{\eps}{16\ell^2}\bigg]\leq \frac{\delta}{4\ell^2}.
\end{align*}
Therefore, with a success probability of at least $1-\delta$ for some $\delta\in (0,1)$ that
\begin{align*}
    \bigg|\sum_{i=0}^{\ell-1}\sum_{j=0}^{\ell-1}\Big(\frac{\bar{W}_{i,j}}{nr}+\frac{\bar{X}_{i,j}}{nr}-\frac{\bar{Y}_{i,j}}{nr}-\frac{\bar{Z}_{i,j}}{nr}\Big)-\frac{(\hh^{\ell}\one)^2}{n}\bigg|\leq\frac{4\eps\ell^2}{16\ell^2}=\frac{\eps}{4}.
\end{align*}

\subsection{Full Description of Algorithm}
We propose \STW, as summarized in Algorithm~\ref{alg:walk}, to answer BD queries by sampling truncated walks. It begins with computing the maximum length $\ell$ and the number $r$ of needed samples. Afterward, \STW starts $r$ sampling of walks of lengths $i$ and $j$ for $i\in [0, \ell]$ and $j\in[0,\ell]$, in each of which it calculates the related quantities (Lines 4--14). Finally, \STW computes $\beta^{\prime}(s,t)$ (Line 15) using the quantities computed from the simulated walks and returns it as the estimated biharmonic distance. The performance of \STW is stated as follows.

\begin{theorem}
    Let $\beta^{\prime}(s,t) = \STW(\calG,s, t, \eps, \delta)$ denote the estimated BD by \STW for some $\delta\in (0,1)$, we have that $|\beta(s,t)-\beta^{\prime}(s,t)|\leq \eps$.
\end{theorem}
\normalem
\begin{algorithm}[t]
	\caption{$\STW(\calG,s,t,\eps,\delta)$}\label{alg:walk}
	\Input{
		A connected graph $\calG=(V,E)$ with $n$ nodes; an error parameter $\eps$, two nodes $s$ and $t$
	}
	\Output{
		Estimated biharmonic distance $\beta^{\prime}(s,t)$
	}
	$\ell=\min(\text{Eq.~\eqref{eq:uni_l}},\text{Eq.~\eqref{eq:cus_l}})$ - 1; $r=\lceil 128 \ell^4\log (\frac{8\ell^2}{\delta}) / \eps^2\rceil$ \;
    $W_{i,j}=0;X_{i,j}=0;Y_{i,j}=0;Z_{i,j}=0;$\;
    $\bar{W}_{i,j}=0;\bar{X}_{i,j}=0;\bar{Y}_{i,j}=0;\bar{Z}_{i,j}=0;$\;
	\For{$i = 0$ to $\ell$; $j = 0$ to $\ell$}{
    \For{${t}=1$ to $r$}{
    Perform two random walks from $s$ (resp., $t$) of length-$i$ and length-$j$, and let $s_1$ and $s_2$ (resp., $t_1$ and $t_2$) denote their end nodes\;
        \lIf{$s_1=s_2$}{$W_{i,j}=W_{i,j}+1$ with probability $1/\dd_{s_1}^2$} 
        \lIf{$t_1=t_2$}{
        $X_{i,j}=X_{i,j}+1$ with  probability  $1/\dd_{t_1}^2$
        }
        \lIf{$s_1=t_2$}{
        $Y_{i,j}=Y_{i,j}+1$ with probability $1/\dd_{s_1}^2$ 
        }
        \lIf{$t_1=s_2$}{
        $Z_{i,j}=Z_{i,j}+1$ with probability $1/\dd_{t_1}^2$ 
        }

        $\bar{W}_{i,j}=\bar{W}_{i,j}+1$ with probability $1/{(\dd_{s_1}\dd_{s_2})}$\;
        $\bar{X}_{i,j}=\bar{X}_{i,j}+1$ with probability $1/{(\dd_{t_1}\dd_{t_2})}$\;
        $\bar{Y}_{i,j}=\bar{Y}_{i,j}+1$ with probability $1/{(\dd_{s_1}\dd_{t_2})}$\;
        $\bar{Z}_{i,j}=\bar{Z}_{i,j}+1$ with probability $1/{(\dd_{t_1}\dd_{s_2})}$
    }}
    $\beta^{\prime}(s,t)\gets\frac{W_{i,j}+X_{i,j}-Y_{i,j}-Z_{i,j}}{r}-\frac{\bar{W}_{i,j}+\bar{X}_{i,j}-\bar{Y}_{i,j}-\bar{Z}_{i,j}}{nr}$\;
    \Return $\beta^{\prime}(s,t)$\;
\end{algorithm}
\section{Sampling with Feedback}\label{sec:walk2}

Utilizing Chernoff-Hoeffding's inequality, the preceding section pre-determined the maximum necessary random walks. However, due to the actual variance of random variables often being lower than estimated, many walks become unnecessary. Ideally, an efficient method would utilize the fewest random walks, determined by the actual, albeit unknown, variance.

Inspired by the preceding discussion, we suggest modifying the sampling phase to sample with feedback. Specifically, rather than executing all random walks at once, we progressively conduct Monte Carlo random walks to estimate \(\beta(s,t)\), terminating based on the empirical error from observed samples, calculable via the empirical Bernstein inequality in the following lemma.

\begin{lemma}\label{lem:bernstein}
    (Empirical Bernstein Inequality~\cite{audibert2007tuning}). Let $Z_1,Z_2,\ldots, Z_{n_z}$ be real-valued i.i.d. random variables, such that $0\leq Z_j\leq\psi$. We denote by $Z=\frac{1}{n_z}\sum_{j=1}^{n_z}Z_j$ the empirical mean of these variables and $\hat{\sigma}^2=\frac{1}{n_z}\sum_{j=1}^{n_z}(Z_j-\frac{Z}{n_z})^2$ their empirical variance. Then, for some $\delta\in (0,1)$, we have 
    $$\mathbb{P}\left[|Z-\mathbb{E}[Z]| \geq f\left(n_z, \hat{\sigma}^2, \psi, \delta\right)\right] \leq \delta,$$
    where $f\left(n_z, \hat{\sigma}^2, \psi, \delta\right)=\sqrt{\frac{2 \hat{\sigma}^2 \log (3 / \delta)}{n_z}}+\frac{3 \psi \log (3 / \delta)}{n_z}.$
\end{lemma}

We next show how we utilize the empirical Bernstein inequality to devise a new algorithm. We first present the following lemma.
\begin{lemma}\label{lem:bound_Z}
    For two length-$\ell$ random walks $W_1$ and $W_2$ that originate from two nodes (be they the same or not) on graph $\calG$, we have $\xi(W_1, W_2)\leq \frac{\ell^2}{\min^2{\dd}}$ and $\xi^\prime(W_1, W_2)\leq\frac{\ell^2}{\min^2{\dd}}$, where $\xi$ and $\xi^\prime$ are defined as, respectively $$\xi(W_1, W_2)= \sum_{x\in W_1}\sum_{y\in W_2} \frac{\mathbb{I}_{x=y}}{\dd_x^2}, \text{and } \xi^\prime(W_1, W_2) = \sum_{x\in W_1}\sum_{y\in W_2} \frac{1}{\dd_x\dd_y}.$$
    
\end{lemma}
Let $S_1$ and $S_2$ (resp. $T_1$ and $T_2$) be two random walks of length-$(\ell-1)$ originating from $s$ (resp. $t$). Consider the random variable $Z_k$ arising from the random walks such that 
\begin{align}
    \label{eq:zk}Z_k &= Z_{k1} - Z_{k2}/n,\\
    \text{where } Z_{k1} &= \xi(S_1, S_2) + \xi(T_1, T_2) - \xi(S_1, T_2) - \xi(S_2, T_1)\notag\\
    \text{and } Z_{k2} &= \xi^\prime(S_1, S_2) + \xi^\prime(T_1, T_2) -  \xi^\prime(S_1, T_2) - \xi^\prime(S_2, T_1)\notag.
\end{align}
By Lemma~\ref{lem:bound_Z}, we have $ |Z_k|\leq \frac{2\ell^2}{\min^2{\dd}} + \frac{2\ell^2}{n\min^2{\dd}} = \psi$.
It is trivial to derive the expectation of $Z_k$ as $\mathbb{E}[Z_k] = \beta^{\ell}(s,t)$, indicating that $Z$ is an unbiased estimator of $\beta^{\ell}(s,t)$.

According to the Chernoff-Hoeffding inequality, for some $\delta\in (0,1)$, the maximum number of random walks is bounded by
\begin{align}
    \label{eq:eta}
    r^*=\frac{\psi^2\log(2/\delta)}{2\eps^2}.
\end{align}
Based on the above analysis, we propose \SWF, as illustrated in Algorithm~\ref{alg:swf}. It begins by computing the truncated length $\ell$ and the number of samples $r^*$. \SWF repeats the process of sampling and computes the empirical mean $M_k$ and variance $\hat{\sigma}^2$. \SWF terminates under two conditions: (i) $\epsilon_f \leq \frac{\epsilon}{2}$, where $\epsilon_f=f\left(\eta, \hat{\sigma}^2, \psi, \delta\right)$ (Line 12), or (ii) $k \geq r^*$ (Line 4). 

The running time of \SWF is bounded by $O(r^*\ell)$ as the length of each random walk in \SWF is $\ell$. In summary, the total time complexity \SWF for answering pairwise BD queries is $O(\frac{2}{\eps^2\min^4\dd}\ell^5)$.

\begin{theorem}
    \SWF returns an estimated BD $\beta^{\prime}(s,t)$ such that $|\beta^{\prime}(s,t) - \beta(s,t)|\leq \eps$ with a probability of at least $1-\delta$.
\end{theorem}

\normalem
\begin{algorithm}[h]
\label{swf}
	\caption{$\SWF(\calG,s,t, \eps, \delta)$}\label{alg:swf}
	\Input{
		A connected graph $\calG=(V,E)$; two nodes $s$ and $t$
	}
	\Output{
		Estimated biharmonic distance $\beta^\prime(s,t)$
	}
    $\ell= \min(\text{Eq.~\eqref{eq:uni_l}},\text{Eq.~\eqref{eq:cus_l}})$- 1; Compute $r^*$ by Eq.~\eqref{eq:eta}\;
    $Z\gets0$; $\hat{\sigma}^2 \gets 0$\;
    \For{$k=1$ to $r^*$}{
    Perform two length-$\ell$ random walks $S_{k,1}$ and $S_{k,2}$ from $s$\;
    Perform two length-$\ell$ random walks $T_{k,1}$ and $T_{k,2}$ from $t$\;
    Compute $Z_k\gets Z_{k1} - \frac{1}{n}Z_{k2}$ as given in Eq.~\eqref{eq:zk}\;
    Update $Z\gets Z + Z_k$\;
    Update $\hat{\sigma}^2\gets \hat{\sigma}^2 + Z_k^2$\;
    Compute $M_k\gets\frac{Z}{k}$\;
    Compute $\hat{\sigma}_k^{2}\gets \frac{\hat{\sigma}^2}{k} - M_k^2$\;
    \lIf{$f\left(k, \hat{\sigma}_k^{2}, \psi, \delta \right)\leq \frac{\eps}{2}$}{\textbf{break}}
    }
    \Return $\beta^\prime(s,t)=M_k$\;
\end{algorithm}

\section{Nodal Biharmonic Distance Query}\label{sec:nbd}
In this section, we develop two novel and efficient algorithms to process the nodal biharmonic distance $\beta(s)$ query for a node $s$. Note that all the algorithms proposed for estimating pairwise BD can be easily extended to handle the nodal query problem by processing $n-1$ queries $\beta(s,t) \forall t\in V\setminus \{s\}$. We select our fastest algorithm \SWF as the base algorithm and call the resulting algorithm for approximating nodal BD as \SNB. 
\begin{theorem}\label{lem:snb}
Given a graph $\calG$, a source node $s$, and an error parameter $\eps$, \SNB runs in time $O(\frac{n}{\eps^2}\log^5\frac{n}{\eps})$ and returns an estimated nodal BD $\beta^\prime(s)$ such that $|\beta(s) - \beta^\prime(s)|\leq n\eps.$
\end{theorem}

However, \SNB's time complexity is nearly $n$ times that of processing a single-pair query, proving costly for large graphs. Fortunately, we can leverage ``summation estimate'' techniques~\cite{feigesum2006,sum2022} to mitigate computational demands. We first present the subsequent lemma.
\begin{lemma}\label{lem:sqrtsum}
Given $n$ bounded elements $x_1,x_2,\ldots,x_n \in [0,a]$, an error parameter $\tau>an^{-1/2}\log^{1/2} n$, we randomly select $t=O(a\sqrt{n(\log n)}/\tau)$ elements $x_{c_1},x_{c_2},\ldots,x_{c_t}$ by Bernoulli trials with success probability $p=an^{-1/2}\log^{1/2} n/\tau$ satisfying $0<p<1$. Then, we have $\bar{x}=\sum_{i=1}^t x_{c_i}/p$ as an approximation of the sum of the original $n$ elements $x=\sum_{i=1}^n x_i$, satisfying $|x-\bar{x}|\leq n\tau$.
\end{lemma}

Based on the above lemma, we develop a stronger version of \SNB, named \SNBp (outlined in Algorithm~\ref{alg:single_source}). 
\begin{theorem}\label{lem:snb+}
    Given a graph, a query node $s$, and an error parameter $\eps$. Let $\phi=\gamma^{-2}$~\cite{wei2021biharmonic} be the upper bound of BD, $\tau=\eps/2$, and $\theta=\eps/2$. Let $t=O(\phi\sqrt{n(\log n)}/\tau)$, $\widetilde{V}=\{x_1,x_2,\ldots,x_t\}$ be a selected subset based on Lemma~\ref{lem:sqrtsum}. \SNBp  returns an estimated nodal biharmonic distance $\beta^\prime(s)=\sum_{t\in \widetilde{V}} \beta^{\prime}(s,t)n/(\phi\sqrt{n(\log n)}/\tau)$ where $\beta^{\prime}(s,t)=\SWF(\calG,s,t, \theta, \delta)$. \SNBp runs in $O(\frac{1}{\eps^3}\phi n^{\frac{1}{2}}\log^6(\frac{n}{\eps}))$ time, satisfying that $|\beta(s)-\beta^\prime(s)|\leq n\eps$. 
\end{theorem}

\normalem
\begin{algorithm}[h]
	\caption{$\SNBp(\calG,s, \eps, \phi, \delta)$}\label{alg:single_source}
	\Input{
		A connected graph $\calG=(V,E)$ with $n$ nodes; a node $s$; the error parameter $\eps$; a number $\phi$
	}
	\Output{
		Estimated nodal biharmonic distance $\beta^\prime(s)$
	}
    $\theta=\frac{\eps}{2}$; $\tau = \frac{\eps}{2}$\;
    Sample the set $\samv \subset V$ of $x=O(\phi\sqrt{n(\log n)}/\tau)$ nodes\;
    $\beta^{\prime}(s,t) = \SWF(\calG,s,t, \theta, \delta) \forall t\in \samv$\;
    $\beta^\prime(s) =\sum_{t\in \samv} \beta^{\prime}(s,t)n/(\phi\sqrt{n(\log n)}/\tau)$\;
    \Return $\beta^\prime(s)$\;
\end{algorithm}

\section{Experiments}\label{sec:exp}

This section evaluates the efficiency and accuracy of our proposed algorithms experimentally.
\subsection{Experimental Setup}
\textbf{Datasets, Query Sets, and Ground-truth.} To evaluate the efficiency and accuracy of our proposed algorithms, namely \texttt{Push}, \texttt{Push+}, \STW, and \SWF for approximating pairwise BD, as well as \SNB and \SNBp for approximating nodal BD, we conduct experiments on 5 real networks from $\mathrm{SNAP}$~\cite{LeSo16}, details of which are provided in Table~\ref{tab:statistics}. 
For each network, we pick 100 node pairs uniformly at random as the query set. The ground-truth BD values for query node pairs are obtained by applying \texttt{Push} with 1000 iterations in parallel ($\epsilon$ is in $(10^{-6},10^{-4})$).

\begin{table}[h]
    \caption{Experimented real-world networks.}
    \label{tab:statistics}
    \fontsize{9}{8}\selectfont
    \begin{tabular}{lrrcc}
    \toprule
    Network   & \#nodes ($n$) & \#edges ($m$) & avg($d$)       \\ \hline
    Facebook   &  4,039 & 88,234& 43.69  \\[2pt]
    DBLP   & 317,080 &1,049,866 &6.62 \\[2pt]
    Youtube   &  1,134,890 &2,987,624 &5.27  \\[2pt]
    Orkut   &   3,072,441 & 117,185,083 & 76.28 \\[2pt]
    LiveJournal   & 3,997,962 &  34,681,189 & 17.35 \\[2pt]
    \bottomrule
    \end{tabular}
\end{table}

\noindent \textbf{Implementation Details.} All experiments are conducted on a Linux machine with an Intel Xeon(R) Gold 6240@2.60GHz 32-core processor and 128GB of RAM. Each experiment loads the graph used into memory before beginning any timings. The eigenvalues $\lambda_2$, $\lambda_n$ and $\gamma$ of each tested network are approximated via ARPACK~\cite{lehoucq1998arpack}. All tested algorithms are implemented in Julia. For all randomized algorithms, we set failure probability $\delta = 0.01$. We report the average query times (measured in wall-clock time) and the actual average absolute error of each algorithm on each network with various $\eps \in \{0.01, 0.02, 0.05, 0.1, 0.2\}$. We exclude a method if it fails to report the result for each query within one day.

\noindent \textbf{Competitors.} We compare our algorithms with the random projection method \texttt{RP}~\cite{YiSh2018,Yi2022BiharmonicDP} and the \texttt{EXACT} method which requires computing the pseudo-inverse of matrix $\LL$.

\begin{figure*}[t!]
    \centering
    \includegraphics[width=0.95\linewidth]{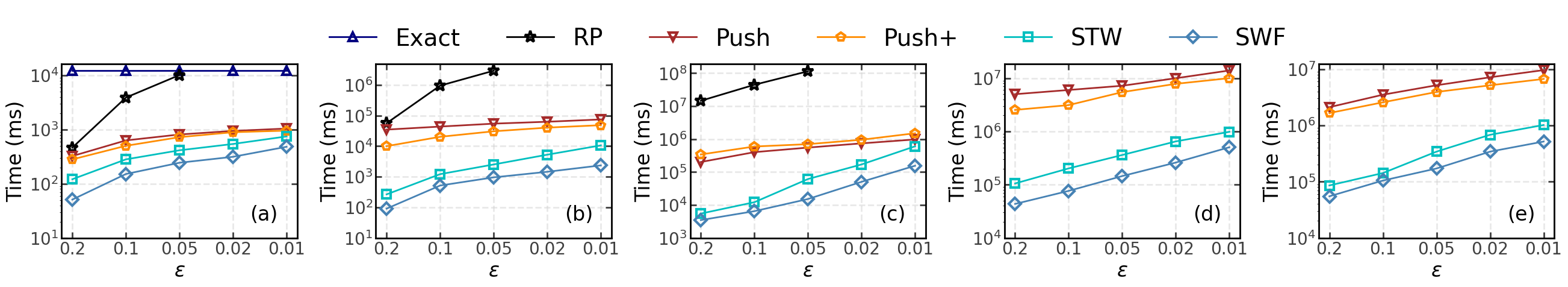}
    \vspace{-0.45cm}
    \caption{Running time of each algorithm on five datasets: (a) Facebook, (b) DBLP, (c) Youtube, (d) Orkut and (e) Livejournal.}
    \label{fig:running_time}
\end{figure*}

\subsection{Query Efficiency}

In our initial series of experiments, we conduct a comparative assessment of the efficiency of our algorithms against others. In Figure~\ref{fig:running_time}, we report the evaluation results on query efficiency (i.e., average running time) of each method on each network when $\epsilon$ is varied from 0.01 to 0.2 ($x$-axis). For \texttt{Exact} and \texttt{RP}, we plot their preprocessing time. Note that the $y$-axis is in log-scale and the measurement unit for running time is millisecond (ms). \texttt{Exact} and \texttt{RP} cannot terminate within one day in some cases, and, thus their results are not reported. Specifically, \texttt{Exact} can only handle the smallest network Facebook as it incurs out-of-memory errors on larger networks due to the space requirements for computing the $n \times n$ matrix pseudo-inverse. Akin to \texttt{Exact}, \texttt{RP} runs out of memory on Orkut and LiveJournal networks, where it requires constructing large dense random matrices. On networks such as Facebook, DBLP, and YouTube, both \texttt{Push} and \texttt{Push+} show a significant acceleration relative to \texttt{RP}. Remarkably, algorithms \STW and \SWF outpace other algorithms, often by a staggering margin. To illustrate, in the YouTube network, when $\eps=0.05$, \STW runs several orders faster than \texttt{RP} and $46$ times faster than \texttt{Push+}.

Both \STW and \SWF significantly outperform other algorithms also on larger networks Orkut and LiveJournal. Notably, on such graphs, the efficiency of \texttt{Push} and $\texttt{Push+}$ is largely degraded on account of the expensive matrix-vector multiplications. Although \STW is faster than \texttt{Push} and $\texttt{Push+}$, it is still relatively costly as it involves numerous random walks on large graphs when $\epsilon$ is small. In comparison, \SWF, due to its feedback-driven sampling strategy, dominates all other algorithms significantly, and even runs up to $23$ times faster than \STW on Orkut network when $\eps=0.05$.

As for the efficiency of \SNB and \SNBp for approximating nodal BD, we select two representative networks Facebook and DBLP. We report the running time in Figure~\ref{fig:nodal} (a) and (b), the proposed algorithms \SNB and \SNBp both outperform \textsc{Exact} by a large margin.  On Facebook, \SNB runs slightly slower than \texttt{RP} in a few cases. However, thanks to the efficient ``summation estimation'' technique, \SNBp's running time is significantly smaller than \SNB as well as \texttt{RP} on both networks under all $\eps$ values.

\begin{figure}[t]
    \centering
    \includegraphics[width=0.92\columnwidth]{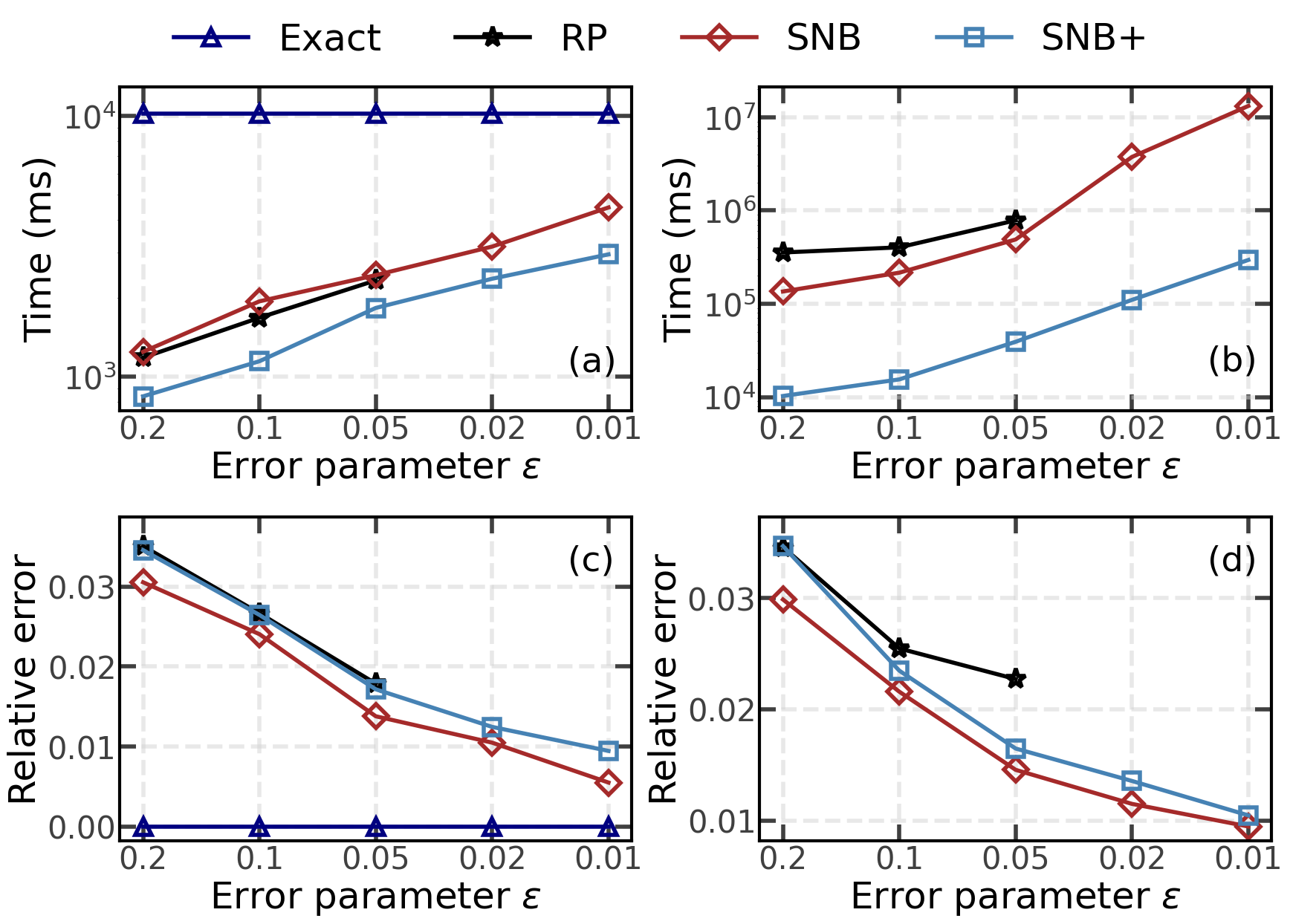}
    \vspace{-0.45cm}
    \caption{Performance comparison for approximating nodal BD on two networks: (a) (c) Facebook, (b) (d) DBLP.}
    \label{fig:nodal}
\end{figure}
\begin{figure}[t]
  \centering
  \includegraphics[width=0.92\columnwidth]{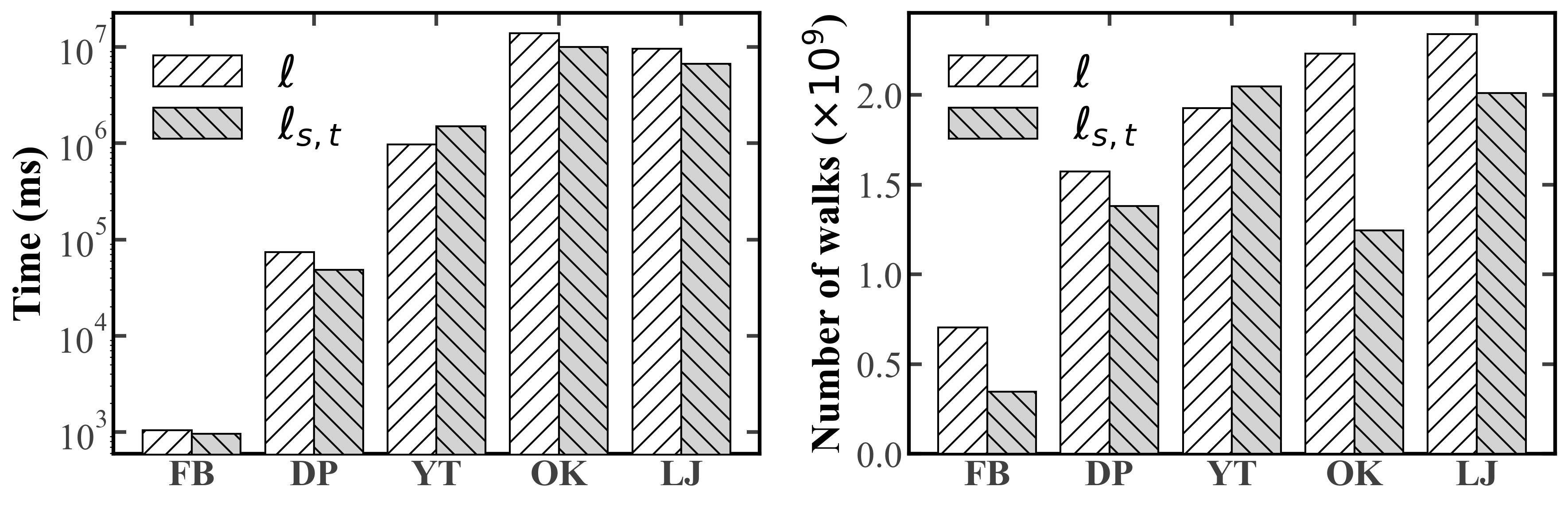}
  \vspace{-0.45cm}
  \caption{$\ell$ versus $\ell_{s,t}$ when $\eps=0.01$.}
  \label{fig:l1l2}
\end{figure}
\boldsymbol{$\mathbf{\ell \textrm{ \textbf{versus} } \ell_{s,t}.}$} We next evaluate the efficiency performance of \STW with the truncated length $\ell$ (Eq.~\eqref{eq:uni_l}) and $\ell_{s,t}$ (Eq.~\eqref{eq:cus_l}), respectively. Figure~\ref{fig:l1l2} (a) reports the running time using $\ell$ and $\ell_{s,t}$ for 100 random pairs of nodes on five networks, namely FB, DP, YT, OK, and LJ for $\epsilon=0.01$. Using $\ell$ obviously needs more time than using $\ell_{s,t}$ on higher average degree graphs FB, DP, OK, and LJ. On the graph with a lower average degree YouTube, using $\ell$ yields to be slower than using $\ell_{s,t}$. This phenomenon is due to that $\ell_{s,t}$ is inversely correlated with the degrees of nodes, meaning that $\ell_{s,t}$ will be much smaller than $\ell$ on graphs with larger average degrees. In summary, $\ell_{s,t}$ can bring considerable efficiency enhancements compared to $\ell$ on graphs with high average degrees. Note that the number of samples relies on the truncated length, and therefore, reducing length can lead to a reduced number of samples. Figure~\ref{fig:l1l2} (b) reports the average number of random walks, where the major overhead of \STW stems from. Akin to the observation from Figure~\ref{fig:l1l2} (a), \STW with $\ell_{s,t}$ requires at most $2$ times fewer random walks than $\ell$ in networks with high average degrees and also requires more random walks than $\ell$ in the low average degree network YouTube.
\begin{figure*}[t]
  \centering
  \includegraphics[width=0.95\linewidth]{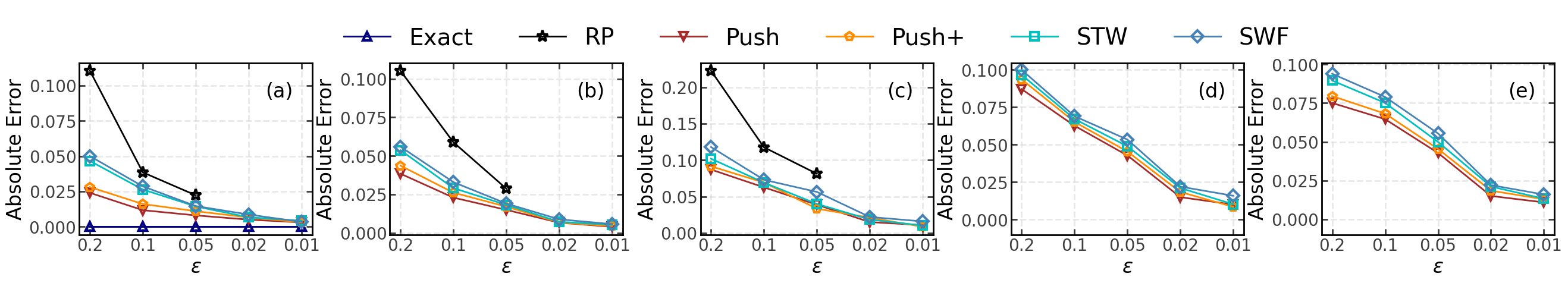}
  \vspace{-0.45cm}
  \caption{Absolute error versus $\eps$ on five datasets: (a) Facebook, (b) DBLP, (c) Youtube, (d) Orkut and (e) Livejournal.}
  \label{fig:abs_err}
\end{figure*}
\subsection{Query Accuracy}
Figure~\ref{fig:abs_err} reports the actual average absolute error for each algorithm on different networks when varying $\epsilon$ from 0.01 to 0.2 for random queries. Note that the $x$-axis and $y$-axis represent the given absolute error threshold $\epsilon$ and the actual average absolute error, respectively. As observed in Figure~\ref{fig:abs_err}, all tested methods return accurate query results, whose actual absolute errors are less than the given error threshold $\epsilon$ (except \texttt{RP}). More specifically, most algorithms achieve average absolute errors smaller than 0.1 even for large error thresholds such as $\epsilon=0.2$ and the errors approach 0 as $\epsilon$ is decreased. The \texttt{RP} algorithm always produces the highest empirical errors. For all networks, \texttt{Push} and \texttt{Push+} achieve similar approximation errors for all $\eps$ settings. This suggests that by selecting an appropriate truncation length ($\ell$ or $\ell_{s,t}$), we've managed to reduce the runtime (see Figure~\ref{fig:running_time}) without incurring a loss in approximation accuracy. As for the two random walk-based algorithms \STW and \SWF, they yield slightly larger errors compared to the competitors \texttt{Push}, and \texttt{Push+}. The reason is that the matrix-vector multiplications lead to a more accurate estimation of BD. Owing to \SWF's sophisticated feedback-driven sampling strategy, there's a reduction in the number of random walks executed, leading to a modest uptick in approximation error relative to \STW. Nonetheless, this marginal trade-off is justifiable, as illustrated in Figure~\ref{fig:running_time}, where \SWF significantly outperforms \STW in terms of efficiency.

As for the query accuracy of \SNB and \SNBp for approximating nodal BD, we report results in Figure~\ref{fig:nodal} (c) and (d). Both of the proposed algorithms \SNB and \SNBp achieve low empirical absolute error and \SNBp presents slightly higher error than \SNB due to the error caused by the ``summation estimation'' technique. They both outperform \texttt{RP} in terms of the running error.

\section{Related Work}\label{sec:related}
In this section, we discuss some prior work on biharmonic distance and explain how they are related to ours.

Biharmonic distance has found applications in different fields. In~\cite{ACC2018BHD, Yi2022BiharmonicDP}, the authors established a connection between the biharmonic distance of a graph and its second-order network coherence. In~\cite{YiSh2018}, the biharmonic distance was used to measure the edge centrality in a network. Furthermore, biharmonic distance has also been utilized in computational graphics~\cite{lipman2010biharmonic,verma2017hunt}, machine learning~\cite{kreuzer2021rethinking,black2023understanding}, and physics~\cite{Tyloo2017RobustnessOS}, as well as graph matching~\cite{fan2020spectral} and leader selection in noisy networks~\cite{bd14}.

Despite its importance, the prior work has made little progress on devising efficient and effective practical algorithms for computing BD. In~\cite{zhang2020fast,Yi2022BiharmonicDP}, the authors developed an algorithm to compute BD whose run time is linear in the number of edges in the network. However, this time complexity is not desirable when dealing with large networks and also their algorithm leads to large hidden memory load. The present work proposes local algorithms (reading only a small portion of the input graph) to approximate biharmonic distance with small space load.

There also exist several local algorithms for other random walk based quantities, such as the stationary distribution~\cite{lee2013computing}, PageRank~\cite{bressan2018sublinear}, Personalized PageRank (PPR)~\cite{wang2017fora} and transition probabilities~\cite{banerjee2015fast}. At first glance, it might seem that we could simply adapt and extend these techniques for BD computation. However, computing BD is much more involved. This is because they are defined according to inherently different types of random walks. More concretely, PPR leverages the one called random walk with restart (RWR)~\cite{Tong2006FastRW}, which would stop at each visited node with a certain probability during the walk. ER relies on simple random walks of various fixed lengths (from 0 to $\infty$)~\cite{peng2021local}. In contrast, BD relies on comparing the end nodes of pairs of walks with various lengths, indicating that the sampling process in BD is more complex. 

\section{Conclusion}\label{sec:conclution}
We provided a novel formula of biharmonic distance (BD) and subsequently proposed using deterministic truncated graph traversal to approximate BD queries. Specifically, we presented \texttt{Push} and \texttt{Push+}, where the former is with a universal truncated length for any pair of nodes while the latter is with a customized length for each query. These two choices suit different scenarios and are adopted according to the setup in hand. For the purpose of practical efficiency, we further developed \STW and \SWF, which are built on the idea of sampling random walks to estimate the desired probabilities. \STW directly samples walks to approximate related probabilities while \SWF observes whether in the process of random walk sampling, the predefined error criterion has been met, thus enabling an early termination of the simulation. Based on the algorithms for approximating pairwise BD, we propose \SNB for computing nodal BD, and further propose \SNBp to speed up its run time. The experimental results on real-world graph data demonstrate that our algorithms consistently and significantly outperform the existing competitors.
In the future, we aim to study the computation of biharmonic distance with relative error guarantees as well as under distributed and multithreading setups. Furthermore, it would be interesting to speed up our proposed algorithms leveraging novel heuristic and approximation techniques.

\newpage

\appendix

\section{Proof of Lemma~\ref{lem:snb}}
Since \SWF returns $\beta^\prime(s,t)$ within $O(\frac{1}{\eps^2}\log^5\frac{n}{\eps})$ time satisfying $|\beta^\prime(s,t) - \beta(s,t)|\leq \eps$ for any node pair $(s,t)$, \SNB actually computes $\beta^\prime(s,t)$ for every $t\in V\setminus\{s\}$ by calling \SWF $n-1$ times. Therefore, the time complexity of \SNB is $n\times O(\frac{1}{\eps^2}\log^5\frac{n}{\eps}) = O(\frac{n}{\eps^2}\log^5\frac{n}{\eps})$. Furthermore, based on the triangle inequality, $|\beta(s) - \beta^\prime(s)|\leq n\eps$ holds.
\section{Proof of Lemma~\ref{lem:snb+}}
\SNBp runs by calling \SWF $x = O(\phi\sqrt{n\log n}/\eps)$ times, so  the time complexity of \SNBp is $O(\frac{1}{\eps^3}\phi n^{1/2}\log^6(n/\eps))$. Furthermore, let $\hat{\beta}(s)$ be the estimated nodal BD by \SNB. Based on Lemma~\ref{lem:sqrtsum} and the triangle inequality, $\left|\beta(s)-\beta^{\prime}(s)\right|=\left|\beta(s)-\beta^{\prime}(s)+\hat{\beta}(s)-\hat{\beta}(s)\right| \leq|\beta(s)-\hat{\beta}(s)|+\left|\hat{\beta}(s)-\beta^{\prime}(s)\right| \leq n \epsilon / 2+n \epsilon / 2=n \epsilon $ holds.
\end{document}